\documentclass[11pt,letterpaper]{article}
\usepackage{amssymb}
\usepackage{graphics}
\usepackage[dvips]{epsfig}
\usepackage{fullpage}
\newtheorem{theorem}{Theorem}
\newtheorem{lemma}{Lemma}
\newtheorem{claim}{Claim}
\newtheorem{corollary}{Corollary}

\newtheorem{definition}{Definition}

\newcommand{\qed}{\hfill $\Box$ \bigbreak}
\newenvironment{proof}{\noindent{\bf Proof.~}}{\qed}

\def\cI{{\cal I}}
\def\cD{{\cal D}}

\def\cT{{\cal T}}
\def\cM{{\cal M}}

\renewcommand{\sp}{\mbox{\rm sup}}
\newcommand{\dist}{\mbox{\rm dist}}
\newcommand{\dfs}{\mbox{\sc dfs}}
\newcommand{\lqa}{\mbox{\rm lqa}}
\newcommand{\weight}{\mbox{\rm weight}}
\newcommand{\parent}{\mbox{\rm parent}}
\newcommand{\LPO}{\mbox{\sc lpo}}
\newcommand{\D}{\mbox{\sc d}}

\begin{document}
%%%%%%%%%%%%%%%%%%%%%%%%%%%%%%%%%%%%%%%

\title{An Optimal Labeling Scheme for Ancestry Queries\thanks{This research is supported in part by the ANR projects ALADDIN and PROSE, and by the INRIA project GANG.}}

\author{
Pierre Fraigniaud\\[1ex]
{\small CNRS and Univ. Paris Diderot}\\{\small\sl pierre.fraigniaud@liafa.jussieu.fr}
\and
Amos Korman\\[1ex]
{\small CNRS and Univ. Paris Diderot}\\{\small\sl amos.korman@liafa.jussieu.fr}
}

\date{}

\maketitle

\begin{abstract}
An {\em ancestry labeling} scheme assigns labels (bit strings) to the nodes of rooted trees
such that ancestry queries between any two nodes in a tree can be answered
merely by looking at their corresponding labels. The quality of an ancestry labeling scheme is measured
by its label {\em size}, that is the maximal number of bits in a label of a tree node.

In addition to its theoretical appeal, the design of efficient ancestry labeling schemes is motivated by applications in web search engines. For this purpose, even small improvements in the label size are important.
In fact, the literature about this topic is interested in the  exact label size rather than just its order of magnitude.
As a result, following the proposal of a simple interval-based ancestry scheme with label size $2\log_2 n$ bits (Kannan et al., STOC~'88), a considerable amount of work was devoted to improve the bound on the size of a label.
The current state of the art upper bound is $\log_2 n + O(\sqrt{\log n})$ bits (Abiteboul et al., SODA~'02) which
is still far from the known $\log_2 n + \Omega(\log\log n)$ bits  lower bound (Alstrup et al., SODA'03).

In this paper we close the gap between the known lower and upper bounds, by constructing an ancestry labeling scheme
with label size $\log_2 n + O(\log\log n)$ bits. In addition to the optimal label size, our scheme assigns the labels in linear time and can
support any ancestry query in constant time.

\end{abstract}

\thispagestyle{empty}
\newpage
\setcounter{page}{1}

%%%%%%%%%%%%%%%%%%%%%%%%%%%%%%%%%%%%%%%
\section{Introduction}
\label{section:Introduction}
%%%%%%%%%%%%%%%%%%%%%%%%%%%%%%%%%%%%%%%

%-----------------------------------------------------------
\subsection{Background}
%-----------------------------------------------------------

In this paper we consider the following problem. Given an $n$-node rooted tree $T$,
label the nodes of $T$ in the most compact way such that given any pair of nodes $u$ and $v$,
one can determine whether $u$ is an ancestor of $v$ in $T$ by
merely inspecting the labels of $u$ and $v$.
The main quality measure used to evaluate  such an {\em ancestry labeling scheme} is the label
 {\em size}, that is,
the maximum number of bits stored in a label of a node, taken over all nodes in all possible $n$-node rooted trees.

Among other things, the above elegant problem is not only of fundamental interest but is also useful for  performance enhancement of
XML search engines.  In the context of this application, each indexed document is a tree, and the labels
of all trees are maintained in the main memory\footnote{Details on XML search engines and their relation to ancestry labeling schemes can be found, e.g., \cite{AAKMT01,AKM01,FK09}.}.
Therefore, even small improvements in the label size are important, and, in fact, the literature about this topic is interested in the exact label size rather than just its order of magnitude (e.g., label size $\frac{3}{2}\log n$ bits is considered significantly
better than label size $2\log n$ bits\footnote{All logarithms in this paper are taken
in base 2.}).

Ancestry schemes which are currently being used by actual systems are variants of the following simple interval-based ancestry labeling scheme \cite{KNR92} (see also \cite{SK85}). Given an $n$-node tree $T$, perform a DFS traversal in $T$ starting at the root,
and provide each node $u$ with a
DFS number $\dfs(u)$ in the range  $[0,n-1]$. Then the label of a node $u$ is simply the interval $I(u)= [\dfs(u),\dfs(v)]$, where
 $v$ is the descendant of $u$ with largest DFS number.
An ancestry query then amounts to an interval containment query between the corresponding labels: a node $u$ is an ancestor of a node $v$  if and only if $I(v)\subset I(u)$.
Clearly, the  label size of this scheme is $2\log n$ bits.

An elegant lower bound of $\log n + \Omega(\log\log n)$ bits on the label size is given in ~\cite{ABR05}.
This lower bound holds even for a very restricted family of trees, each composed of equal length simple paths
hanging down from the root.

In the other direction,
 a considerable amount of research has been devoted to improve the upper bound on the label size as much as possible beyond the trivial $2\log n$ bound \cite{AAKMT01,AKM01,ARSODA02,FK09,KMS02,TZ01}. Specifically,  \cite{AKM01} gave a first non-trivial upper bound of $\frac{3}{2}\log n + O(\log\log n)$ bits.
 This was improved the year after to $\log n + O(\sqrt{\log n})$ \cite{ARSODA02}, which is the current best upper bound (that scheme is described in detail in the joint journal publication \cite{AAKMT01}). In addition to its relatively small label size, the scheme in \cite{ARSODA02} also assigns labels in linear time and can answer any ancestry query in constant time. Independently of that work, an ancestry labeling scheme with larger label size of $\log n +O(\log n/\log\log n)$ was given in \cite{TZ01}.

Following the above results, two other works were published, which focused on particular types of trees.
Specifically, an experimental comparison of different ancestry labeling schemes  on XML trees that appear in real life can be found in \cite{KMS02}.
Recently, \cite{FK09} gave an ancestry labeling scheme which is efficient for trees of small depth;
specifically, for $n$-node trees with depth $d$, their scheme uses labels of size $\log n +2\log d +O(1)$.

%-----------------------------------------------------------
\subsection{Our results}
%-----------------------------------------------------------

In this paper we close the gap between the known lower and upper bounds, by constructing an ancestry labeling scheme
for general rooted $n$-node trees with label size $\log n + O(\log\log n)$.
This solves one of the main open problems in the field of informative labeling schemes.
In addition to the optimal label size,  our scheme assigns the labels to the nodes of a tree in linear time, and  guarantees that any ancestry query can be answered in constant time.

%-----------------------------------------------------------
\subsection{Related work}
%-----------------------------------------------------------
As explained in \cite{KNR92}, the names of nodes  in traditional graph representations reveal no information
about the graph structure and hence memory is wasted.
Moreover, typical representations are
usually global in nature, i.e.,  in order to derive  useful information,
one must access a global data structure representing the entire
network, even if the sought information is local, pertaining to only
a few nodes.
In contrast, the notion of {\em informative labeling schemes}, introduced in
\cite{KNR92}, involves an informative method for assigning labels to nodes.
Specifically, the assignment is made in a way that allows one to infer
information regarding any two nodes directly from their labels,
without using any additional information sources. Hence in
essence, this method bases the entire representation on the set of
labels alone.
This method was illustrated in \cite{KNR92}, by giving two elegant and simple  labeling schemes
 for  $n$-node trees: one supporting adjacency queries and the other supporting ancestry queries. Both schemes incur $2\log n$ label size.

As mentioned earlier, ancestry labeling schemes were further investigated in
 \cite{AAKMT01,ABR05,AKM01,ARSODA02,FK09,KMS02,TZ01}, and the current state of the art upper and lower bounds
 are $\log n+O(\sqrt{\log n})$ and $\log n +\Omega(\log\log n)$, respectively.
Adjacency labeling schemes on trees were also further investigated in an
attempt to optimize the label size beyond the simple $2\log n$ bound of \cite{KNR92}.
The current state of the art  upper bound \cite{AR02+}  for that problem is $\log
n+O(\log^* n)$.

Labeling schemes were also proposed for other  decision problems on graphs,
including distance
\cite{ABR05,GPPR01,T01},
routing \cite{FG01,K08,TZ01}, flow \cite{KK06,KKKP04}, vertex connectivity \cite{K07,KKKP04}, nearest common ancestor
\cite{AGKR01,Peleg00:lca}, and various other tree functions, such as
center, separation level, and Steiner weight of a given subset of
vertices \cite{Peleg00:lca}.
See \cite{GP01b} for a partial survey on labeling schemes.

%%%%%%%%%%%%%%%%%%%%%%%%%%%%%%%%%%%%%%%
\section{Preliminaries}
%%%%%%%%%%%%%%%%%%%%%%%%%%%%%%%%%%%%%%%

Let $T$ be a tree rooted at some node $r$ referred as the {\em root} of $T$.
For two nodes $u$ and $v$  in $T$,
we say that $u$ is an {\em ancestor} of $v$ if $u\neq v$ and $u$ is one of the nodes on the shortest path connecting
$v$ and $r$ in $T$.  For every non-root node $u$, let $\parent(u)$ denote the parent of $u$, i.e., the ancestor of $u$ at distance~1 from it. A node $v$ is a {\em descendant} of $u$ if and only if $u$ is an ancestor of $v$.

The {\em depth} of a node $u\in V(T)$
is defined as the distance from $u$ to the root of $T$, i.e., the number of edge traversals from $u$ to the root.
In particular, the depth of the root is 0.
 The {\em size} of $T$, denoted by $|T|$, is
the number of nodes in $T$.
The {\em weight} of a node $u\in V(T)$, denoted by $\weight(u)$,
is defined as 1 plus the number of descendants of $u$, i.e., $\weight(u)$ is the size of the subtree hanging down from $u$.
In particular, the weight of the root is  $\weight(r)=|T|$.
Let $\cT(n)$ denote the family of all rooted trees of size at most $n$.

An {\em ancestry labeling
scheme} $( \cM,\cD )$  for the family of trees $\cT(n)$ is
composed of the following two components:
\begin{enumerate}
\item A {\em marker} algorithm $\cM$ that, given
a tree $T\in \cT(n)$,
assigns labels (i.e., bit strings) to its nodes.
\item A  {\em
decoder} algorithm $\cD$ that, given any two labels $\ell_1$ and $\ell_2$ in the output domain of $\cM$, returns a boolean $\cD(\ell_1,\ell_2)$.
\end{enumerate}

These components must satisfy that if $L(u)$ and $L(v)$ denote the labels assigned by the marker to
two nodes $u$ and $v$ in some rooted tree $T\in T(n)$, then
\[
\cD(L(u),L(v))=1 \iff \mbox{$u$ is an ancestor of $v$ in $T$.}
\]
It is important to note that the decoder $\cD$ is independent of the tree $T$. That is, given the labels of two nodes,
the decoder decides the ancestry relationship between the corresponding nodes without knowing to which  tree in $\cT(n)$  they belong to.

The common complexity measure used to evaluate the quality of an ancestry labeling scheme
$(\cM,\cD)$ is the {\em label size}, that is
the maximum number of bits in a label assigned by the marker
algorithm $\cM$ to any node in any  tree $T\in\cT(n)$.

When considering the {\em query time} of the decoder, we use the RAM model of computation, and assume that the length of a computer
word is $\Omega(\log n)$ bits.
Similarly to \cite{ARSODA02}, our decoder algorithm uses only the basic and fast RAM operations such as  addition, substraction, left/right shifts and less-than comparisons.
Our scheme avoids the sometimes more costly operations such as multiplication, division or non-standard operations which are pre-computed
and stored in a pre-computed table.

\paragraph{Notations.}
For every two nodes $v$ and $w$ in $T$, let $P[v,w]$ denote the shortest path connecting $v$ and $w$ in the tree (including $v$ and $w$),
and let $P[v,w)=P[v,w]\setminus\{w\}$.

For two integers $a\leq b$, let $[a,b]$ denote the set of integers
$\{a,a+1,\cdots, b\}$. We refer to this set as an {\em interval}.
For two intervals $I=[a,b]$ and $I'=[a',b']$, we say that $I\prec I'$ if $b<a'$.
The {\em size} of an interval $I=[a,b]$, denoted by $|I|$, is the number of integers
in $I$, i.e, $|I|=b-a+1$.

%%%%%%%%%%%%%%%%%%%%%%%%%%%%%%%%%%%%%%%
\section{Modifying the interval containment test}
%%%%%%%%%%%%%%%%%%%%%%%%%%%%%%%%%%%%%%%

Our scheme is inspired by the scheme in \cite{FK09} which was designed for trees of bounded depth.
Given a rooted tree $T$,  the label assigned to each node by the scheme in \cite{FK09}  is a pointer to some interval, and an
ancestry query between any two given nodes is answered by a simple interval containment test between the corresponding intervals.
The underlying idea of that scheme consists in proving that, if $T$ is of small depth, then one can choose the intervals from a small set  of intervals $U$ in which the intervals are
well nested within themselves. A pointer to an interval in $U$ can be encoded using $\log |U|$ bits, and thus, since $U$ is relatively small,
the scheme uses short labels. Unfortunately, this approach is no longer efficient when the tree has long paths. Indeed, in that case, the set $U$ of nested intervals becomes too large.

Informally, enforcing the decoder to be merely an interval containment test imposes a strong constraint on the way the intervals
must be organized in $U$. For arbitrary trees, we could not find a way to bypass this constraint while keeping $U$ small.
Instead, we introduce a decoder which, on the one hand, makes the ancestry test somewhat more complicated than when using the interval containment test (yet the test remains very simple), but, on the other hand, enables to organize and nest the intervals in such a way that labels become very small.
Our new decoder exploits the fact that intervals can be partially ordered not only by
the containment relation, but also  by the relation $\prec$ introduced in the previous section.

Given any node $u$ of some rooted tree $T$, we associate $u$ with an interval $I(u)$, and with
a {\em supervisor} node, denoted by $\sp(u)$, which is either $u$ itself or one of its ancestors. For this purpose, we first mark each node  as either {\em heavy} or {\em light} as follows.
For every non-leaf node $u$ of $T$, let $H(u)$ be the set of children $v$ of $u$ that satisfy
$\weight(v)\geq \weight(w)$ for every child $w$ of $u$. Among the nodes in $H(u)$, select an arbitrary node, and call it  {\em heavy}.
A node which is not heavy  is called {\em light}. (In particular, the root is light).

\begin{figure}[t]
\begin{center}
\includegraphics[width=0.3\linewidth]{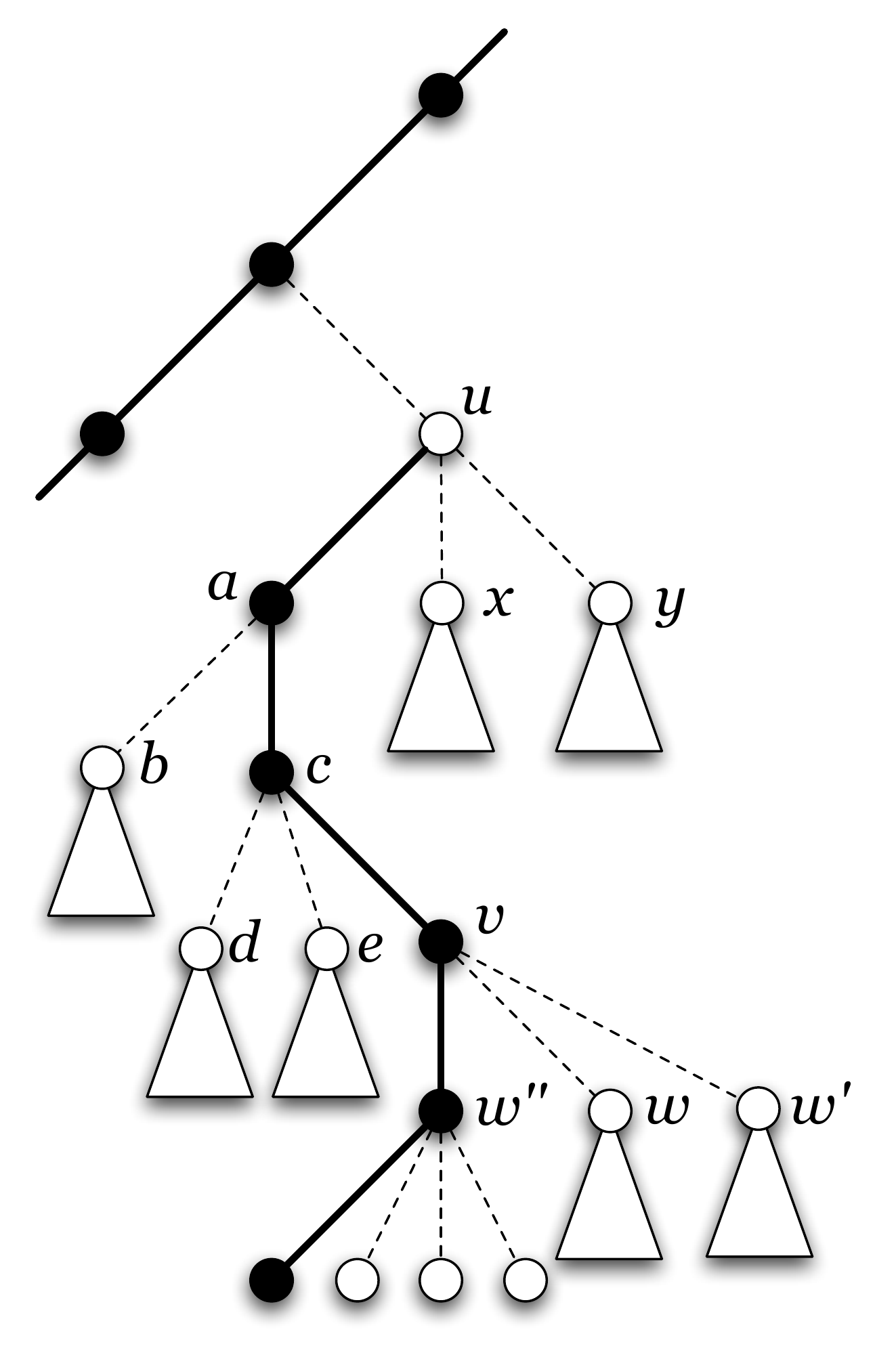}
\caption{The heavy nodes are depicted in black, while the light nodes are depicted in white. In the figure, $\sp(v)=u$ and $\sp(w)=w$. We have $\parent(w)=v$. Hence, the local quasi ancestors of $w$ are the set of nodes is $\lqa(w)=\{x,y,a,b,c,d,e,v,w'\}$ if $\dfs(w)>\dfs(w')$, and
$\lqa(w)=\{x,y,a,b,c,d,e,v\}$, otherwise. Similarly, we have $\lqa(x)=\emptyset$ if $\dfs(y)>\dfs(x)$, and $\lqa(x)=\{y\}$ otherwise. }
\label{fig:suplocheavy}
\end{center}
\end{figure}

For each node $u\in T$, define the {\em supervisor} of $u$, denoted by $\sp(u)$, as the light node of largest depth on the  path $P[u,r]$ connecting
$u$ to the root $r$. Note that if $u$ is light then $\sp(u)$ is $u$ itself; in particular, $\sp(r)=r$, and $\sp(\sp(u))=\sp(u)$. Observe also that if $u$ is an ancestor of $v$, then either $\sp(u)=\sp(v)$, or $\sp(u)$ is an ancestor of $\sp(v)$. See Figure~\ref{fig:suplocheavy}.

As we will show later, the basic rule of our decoder relies on the following definition which is a modification of the interval containment test used in several previous schemes.

\begin{definition}\label{def:decodingcond}

Let us consider a set of intervals $\{I(u),\; u\in V\}$ for a tree $T$. We assume that all intervals in the set are distinct, i.e., $I(u)\neq I(v)$ for any two distinct nodes $u$ and $v$. We say that the {\em decoding} conditions  hold at $u$ w.r.t. $v$ if and only if
\begin{itemize}
\item \D1: $I(v)\subset I(\sp(u))$, and
\item \D2: $I(u)\prec I(v)$ or $I(u)= I(\sp(u))$.
\end{itemize}
\end{definition}

Given the labels $L(u)$ and $L(v)$ of two nodes in a rooted tree,
our boolean decoder $\cD$ outputs~1 if  and only if the decoding conditions hold at $u$ w.r.t. $v$.
Our marker algorithm will then guarantee the following:

\begin{itemize}
\item The decoder is correct, i.e., the intervals associated with each node
are selected such that, for any two nodes $u$ and $v$,  the decoding conditions hold at  $u$ w.r.t. $v$  if and only if $u$ is an ancestor of $v$;
\item given a tree $T$, the interval and labels  can be assigned to all nodes in $T$ in linear time;
\item given a label $L(u)$, the intervals $I(u)$ and $I(\sp(u))$ can be computed in constant time;
\item  each label is encoded using $\log n+O(\log\log n)$ bits.
\end{itemize}

%%%%%%%%%%%%%%%%%%%%%%%%%%%%%%%%%%%%%%%
\section{The $\log n+O(\log\log n)$  ancestry labeling scheme}
%%%%%%%%%%%%%%%%%%%%%%%%%%%%%%%%%%%%%%%

We are now ready to prove our main result, that is:

\begin{theorem}\label{theo:main}
There is an ancestry labeling scheme for $\cT(n)$ with label size $\lceil\log n\rceil+6\lceil\log\log n\rceil+7$ and constant query time.
Moreover, given a tree $T$, the labels can be assigned to the nodes of $T$ in linear time.
\end{theorem}

We prove Theorem~\ref{theo:main} by constructing
an ancestry labeling scheme $(\cM,\cD)$ with the desired properties.

%-----------------------------------------------------------
\subsection{The marker algorithm $\cM$}
%-----------------------------------------------------------
For simplicity of presentation assume that $n$ is a power of 2, and let us fix a tree $T\in\cT(n)$.
Our marker algorithm $\cM$ first assigns an interval to each node in a way such that $u$ is an ancestor of $v$ if and only the decoding
conditions hold at $u$ w.r.t. $v$. For this purpose, we first show that it is sufficient to provide an assignment of intervals that satisfies a more ``local" condition.

% - - - - - - - - - - - - - - - - - - - - - -
\subsubsection{The local partial order conditions}
% - - - - - - - - - - - - - - - - - - - - - -

%We define the notion of \emph{local quasi-ancestors}.
Let us first assign numbers from 0 to $n-1$ to the nodes according to a DFS traversal that starts at the root, and visits light children first. We denote by $\dfs(u)$ the DFS number of $u$.  Let $P_u=P[\parent(u),\sp(\parent(u)]$. We define the \emph{local quasi-ancestors} of $u$, denoted by $\lqa(u)$, as all nodes in $P_u$, together with their light children, but removing $u$, $\sp(\parent(u))$, and all nodes that have DFS numbers higher than $u$. See Figure~\ref{fig:suplocheavy} for an example. Note that the local quasi-ancestors of a node may not form a connected subtree of $T$.

\begin{definition}
Let us consider a set of pairwise distinct intervals $\{I(u),\; u\in V\}$ for a tree $T$.
For every node~$u$, we say that $u$ satisfies the  {\em local partial order (\LPO)} conditions if the two conditions below are satisfied:
\begin{itemize}
\item \LPO1: $I(u)\subseteq I(\sp(u)) \cap I(\sp(\parent(u)))$ for every non root node $u$;
\item \LPO2: $I(x)\prec I(u)$ for every local quasi-ancestor $x\in \lqa(u)$.
\end{itemize}
\end{definition}

\begin{claim}\label{claim:lpo1}
If every node $u$ satisfies the  local partial order condition \LPO1, then for any light node $u$, and any descendent $v$ of $u$, we have $I(v)\subset I(u)$.
\end{claim}

\begin{proof}
The claim is established by induction on the distance between $u$ and $v$. If $\dist(u,v)=1$ then the claim holds by \LPO1. Assume that the claim holds for $\dist(u,v)\leq d$ for $d\geq 1$, and assume $\dist(u,v)=d+1$. If $\sp(v)=u$ then the claim follows by \LPO1. Thus assume  $\sp(v)\neq u$. If $\sp(\parent(v))=u$ then again the claim follows by \LPO1. Thus assume also $\sp(\parent(v))\neq u$. In this case there exists a light node $w\notin\{u,v\}$ on the shortest path connecting $u$ to $v$. The claim then follows by induction.
\end{proof}

The following claim relates the decoding conditions to the local partial order conditions.

\begin{claim}\label{claim:equiv}
If every node $u$ satisfies the  local partial order conditions \LPO1 and \LPO2, then for every two different nodes $u$ and $v$, $u$ is an ancestor of $v$ if and only if the decoding conditions \D1 and \D2 hold at $u$ w.r.t. $v$.
\end{claim}

\begin{proof}
Assume that every node $u$ satisfies the local partial order conditions. Consider first the case that $u$ is an ancestor of $v$.  Since $v$ is a descendent of $u$, either $\sp(v)=\sp(u)$ or $\sp(v)$ is a descendent of $\sp(u)$. Thus, by Claim~\ref{claim:lpo1}, $I(v)\subseteq I(\sp(v)) \subseteq I(\sp(u))$. Since $v\neq u$,   $I(v)\subset I(\sp(u))$, i.e., \D1 follows. If $u$ is light then the fact that $u$ and $v$ satisfy \D2 follows trivially from the fact that, in this case, $\sp(u)=u$. So assume now that $u$ is heavy. If $u\in \lqa(v)$, then \D2 follows from \LPO2. Otherwise, if $u\notin\lqa(v)$, then there exists a light node $w$ that is an ancestor of $v$ and such that $u\in\lqa(w)$. \D2 follows by combining \LPO2 with Claim~\ref{claim:lpo1}.

Consider now the case that $u$ is not an ancestor of $v$. We need to show that either \D1 or \D2 does not hold.
%First, assume that
%$v$ is an ancestor of $u$.
%First, if $v$ is an ancestor of $u $ then we know by the argument above that
% $I(u)\subset I(\sp(v))$, and  $I(v)\prec I(u)$ or $I(v)= I(\sp(v))$.
Let $w$ be the light node of largest depth on the path $P[v,r]$ which is an ancestor
of $u$.
If $w$ is $v$ itself then  $\sp(u)$ is either $v$ or a descendant of $v$. Therefore, by Claim \ref{claim:lpo1}, we have
$I(\sp(u))\subseteq I(v)$, and thus \D1 is not satisfied.
In the remaining proof we thus assume that $w\neq v$.

For a node $x$ which is a descendant of $w$, and which satisfies $\sp(x)\neq w$, let $f(x)$ be
the light node of smallest depth on the path  $P[x,w)$.
Assume first that
$\sp(u)=w$. If also
$\sp(v)=w$ then $v\in \lqa(u)$, and thus, by \LPO2, $I(v)\prec I(u)$. Since $u\neq w$, we have  $I(u)\neq I(\sp(u))$, and therefore
\D2 is not satisfied.
If, on the other hand, $\sp(v)\neq w$ then we have $f(v)\in \lqa(u)$ and $I(v)\subset I(f(v))$. Similarly to the previous case, this implies
that \D2 is not satisfied.

Assume now that  $\sp(u)\neq w$.
We have $I(\sp(u))\subseteq I(f(u))$. If $v$ is an ancestor of $f(u)$ then $v\in \lqa(f(u))$. Consequently, $I(v)\prec I(f(u))$
and thus \D1 does not hold. On the other hand, if $v$ is not an ancestor of $f(u)$ then we are left with two cases:
$\sp(v)=w$ and $\sp(v)\neq w$. If $\sp(v)=w$ then $I(f(u))\prec I(v)$ since $f(u)\in \lqa(v)$. Thus
\D1 does not hold. Finally, if $\sp(v)\neq w$ then  either $f(u)\in \lqa(f(v))$ or $f(v)\in \lqa(f(u))$.
Since $I(\sp(u))\subseteq I(f(u))$ and $I(v)\subseteq I(f(v))$, it follows that \D1 does not hold.
\end{proof}

% - - - - - - - - - - - - - - - - - - - - - -
\subsubsection{The interval assignment}
% - - - - - - - - - - - - - - - - - - - - - -

By Claim~\ref{claim:equiv}, one of our goals is to let the marker assign  intervals that satisfy the local partial order conditions at each node.
For integers $a,b$ and $k$, let
$$ I_{k,a,b}=[2^{k}a,\; 2^{k}(a+b)].$$
For $k\in [1,\log n]$, define the set of intervals:
\[
\cI_k=\{I_{i,a,b}\mid i\in [1,k], a\in [1,\frac{4n\log n}{2^{i}}], \mbox{~and~} b \in [1,4\log n] \}.
\]
Let  $\cI=\cI_{\log n}$.

\begin{definition}
Let $T\in\cT(n)$. We say that a mapping $I:V \rightarrow \cI$ is a {\em legal interval-mapping} if
the mapping is one-to-one, and $\{I(u), u\in V\}$ satisfies the local partial order conditions at each node of $T$.
\end{definition}

In order to show that there exists a legal interval-mapping from every tree in $\cT(n)$
into $\cI$, we use the following notation.
For any interval $J\subset [1,n\log n]$, and, for any $k$, $1\leq k \leq \log n$, let
$$\cI_k(J)= \left\{I_{i,a,b}\in \cI_k\mid  I_{i,a,b} \subseteq J\right\}.$$

\begin{lemma}\label{lemma:main}
For every $k\in [1,\log n]$, every tree $T\in\cT(2^{k})$, and every interval $J\subseteq [1,4n\log n]$,
such that $|J|= 4k|T|$,
there exists an legal interval-mapping of $T$ into $\cI_k(J)$. Moreover this mapping can be computed in $O(|T|)$ time.
\end{lemma}

\begin{proof}
We prove the lemma by induction on $k$.
For $k=1$, the lemma holds trivially.
Assume now that the claim holds for $k$ with $1\leq k< \log n$, and let us show that it also
holds for $k+1$. Clearly, if $|T|\leq 2^{k}$ then we are done by induction.

Consider now the case where $T$ is of size $2^k <|T|\leq 2^{k+1}$, and let $J\subset [1,4n\log n]$ be an interval, such that $|J|=4(k+1)|T|$. Our goal is to show that
there exists a legal-interval mapping of $T$ into $\cI_{k+1}(J)$.

We make use of  the following decomposition of $T$.
Let $H$ be the path from the root of $T$ to a leaf of $T$ such that every non-root node in $H$ is heavy.
 Let $v_1,v_2,\cdots,v_{d}$ be the nodes of  $H$, ordered top-down, i.e.,  $v_1$ is the root of $T$, $v_{d}$ is a leaf of $T$,
 and for every $1\leq i<d$, $v_i$ is the parent of $v_{i+1}$.

For every $1\leq i\leq d$, let $T_i^1,T_i^2,\cdots, T_i^{t_i}$ be the rooted trees hanging down from the light children of $v_i$.
(If $v_i$ does not have any light child, which is the case, for example, for $i=d$, then this set of trees is empty, or, in other words, $t_i=0$).
One important property of these trees is that, for every  $i$ and $j$, $1\leq j \leq t_i$, we have $|T_i^j|< |T|/2\leq 2^{k}$.

\begin{figure}[t]
\begin{center}
\includegraphics[width=0.5\linewidth]{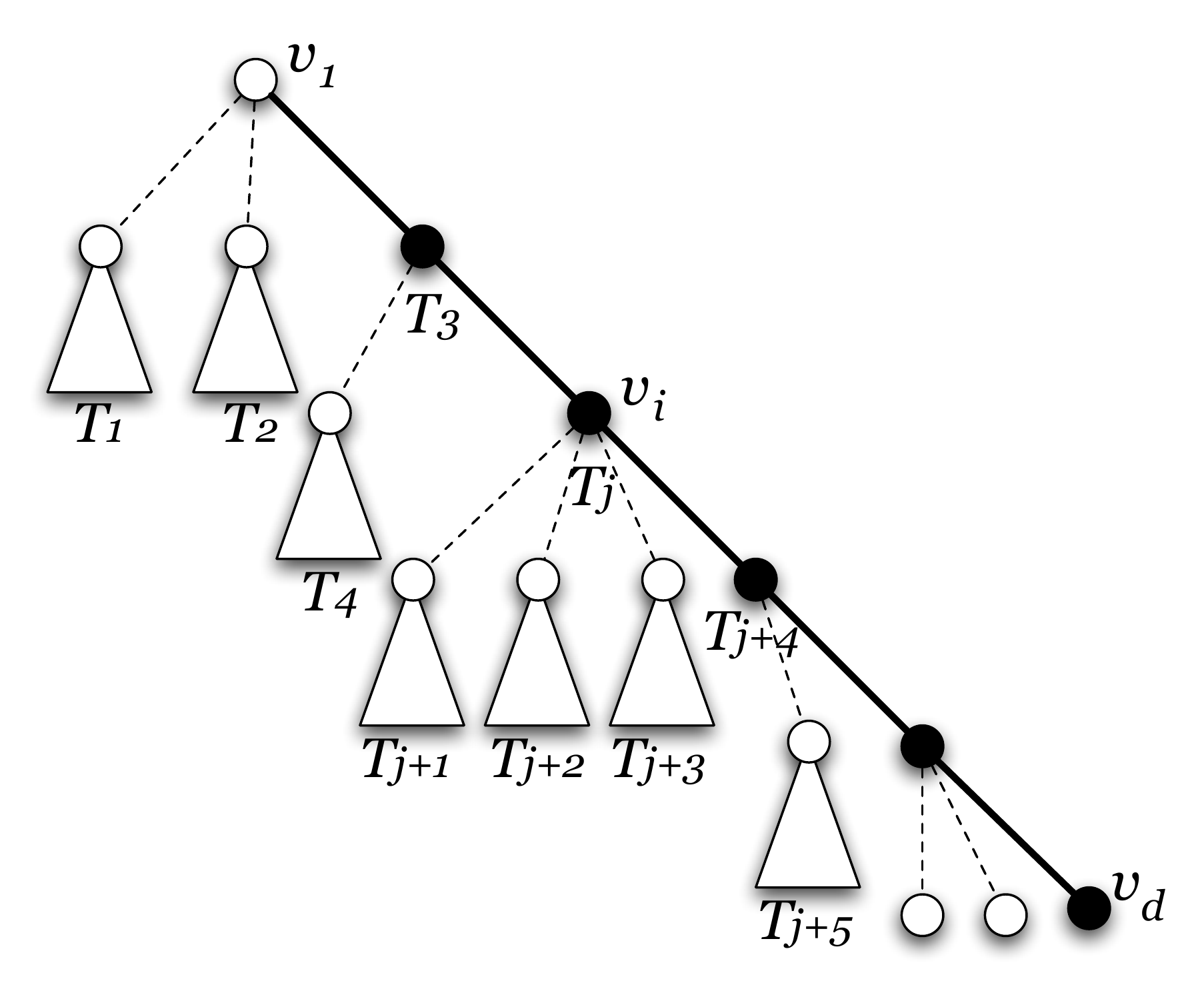}
\caption{Tree decomposition as in the proof of Lemma~\ref{lemma:main}. In the figure, the DFS traversal that visits light nodes first is supposed to proceed by visiting children from left to right.}
\label{fig:recur}
\end{center}
\end{figure}

We now group the nodes in $T\setminus\{r\}$ in disjoint trees $T_1,T_2,\cdots, T_m$, where $m=(d-1)+\sum_{\ell=1}^d t_\ell$ as follows. A tree $T_i$ is either a single heavy node $v_j$, for $j>1$, or a subtree hanging from a light child of some $v_j$, $j\geq 1$. Moreover, the trees are enumerated according to the DFS numbers of their roots as follows. Recall $\dfs(u)$ denotes the DFS number of node $u$ in $T$. The trees are ordered such that if $r_j$ denotes the root of $T_j$ then $\dfs(r_j)<\dfs(r_{j+1})$ for all $j=1,\dots,m-1$. See Figure~\ref{fig:recur}.

Consider now the interval $J\subseteq [1,4n\log n]$ such that $|J|=4(k+1)|T|$, and
 express it as $J=[\alpha, \alpha+4(k+1)|T|-1]$ for some integer $\alpha\leq 4n\log n$.
Let $a$ be the smallest integer such that $\alpha \leq a\, 2^{k+1}$, and
let $b$ be the smallest integer such that $ 4k|T|\leq b\, 2^{k+1}$.

First, we assign the root $r$ to the interval $J'=[a2^{k+1},(a+b)2^{k+1}]$. We now show that indeed $J'\in \cI_{k+1}(J)$. By definition of $a$ and $b$, we have
\[(a+b)2^{k+1}=2^{k+2}+((a-1)+ (b-1))2^{k+1}\leq 2^{k+2}+\alpha+4k|T|-2. \]
Since $2^k<|T|$, we get  that
\[(a+b)2^{k+1}< \alpha+(4k+4)|T|-1. \]
Thus
\[J' = [a\, 2^k,\,(a+b)2^{k+1}]\subset [\alpha,\alpha+(4k+4)|T|-1]= J. \]
Therefore,
since $1\leq a \leq \alpha/2^{k+1}\leq 4n\log n/2^{k+1}$, and $1\leq b \leq 4\log n$, we obtain that
$J'\in \cI_{k+1}(J)$.

The rest of the nodes in $T$ are mapped as follows. First note that $|J'|\geq 4k|T|$, and recall that $\sum_{i=1}^m |T_i|=|T|-1$.
We break $J'$ into $m+1$ consecutive intervals $J'_1,J'_2, \cdots J'_{m+1}$ such that,
for each $1\leq i \leq m$, we have $|J'_i|=4k|T_i|$ and $J'_i\prec J'_{i+1}$.
For each $1\leq i \leq m$, since $|T_i|\leq 2^k$, we can use the induction hypothesis to map the nodes in $T_i$ to $\cI_k(J'_i)$ via
a legal interval-mapping.

The fact that the above recursive mapping can be performed in linear time is obvious.
It remains to show that the above mapping of $T$ into $\cI_{k+1}(J)$ is indeed a legal interval-mapping. That is, we have to show that the set of intervals satisfies the local partial order conditions \LPO1 and \LPO2 at each node $u$. The conditions hold trivially at the root $r$ of $T$.  The fact that the conditions hold for every node $u$ in $T_i\setminus r_i$, follows from the induction hypothesis, and because both $\sp(u)$, $\sp(\parent(u))$, and $\lqa(u)$ are all contained in $T_i$. Finally, consider the root $r_i$ of $T_i$. (Note that if $r_i$ is heavy then $T_i=\{r_i\}$).  \LPO1 holds trivially for $r_i$ because $J'$ contains $J_i$, and, by induction, the interval assigned to $r_i$ is contained in the interval $J'_i \subset J'$. To establish that \LPO2 holds, first observe that $\lqa(r_i)=\{r_1,\dots,r_{i-1}\}$. On the other hand, for every $j=1,\dots,i-1$,  $\dfs(r_j)<\dfs(r_i)$, and thus $J'_j \prec J'_i$. Hence \LPO2 holds for $r_i$ as well. Our mapping is thus
a legal interval-mapping of $T$ into $\cI_{k+1}(J)$. This completes the proof of the lemma.
\end{proof}

By taking $k=\log n$ in the above lemma, we obtain the following.

\begin{corollary}\label{corollary}
Let $T\in\cT(n)$.
There exists a legal interval-mapping of $T$ into $\cI_{\log n}([1,4n\log n])\subset \cI$.
\end{corollary}

% - - - - - - - - - - - - - - - - - - - - - -
\subsubsection{The label assignment}
% - - - - - - - - - - - - - - - - - - - - - -

We are now ready to describe the label $L(u)$ assigned to every node $u$ by our marker algorithm $\cM$.
Given a rooted tree $T\in \cT(n)$, the marker $\cM$  first marks each node as either  heavy or light, and then assigns the $\dfs$
numbers. This clearly takes linear time. Then  the marker maps the nodes of $T$ into $\cI$ using the legal-interval mapping
 given in Corollary~\ref{corollary}. Again, this step takes linear time.

Given a node $u$, the marker uses the first $\log n +3\lceil\log\log n\rceil+3$ least significant bits of $L(u)$ to encode the interval $I(u)$.
 This can be done explicitly as $I(u)$ is of the form $I_{i,a,b}$ for some $i\in [1,\log n]$,   $a\in [1,\frac{4n\log n}{2^{i}}]$ and $b\in [1,4\log n]$.

Finally, the marker aims at encoding $I(\sp(u))$ in the label of $u$. However, using the method above to encode $I(\sp(u))$ would consume yet another $\log n +3\lceil\log\log n\rceil+3$ bits, which is obviously undesired. Instead, we use the following trick.
Let $i',a'$, and $b'$ be such that
$I(\sp(u))=I_{i',a',b'}$. (Note that if $\sup(u)=u$ then we simply have $i'=i$, $a'=a$ and $b'=b$). Clearly, $2\lceil\log\log n\rceil +2$ bits suffice to encode both  $i'$ and $b'$.
To encode $a'$, the marker acts as follows. Let $I(u)=[\alpha,\beta]$, and let $a''$ be the largest integer such that $2^{i'} a''\leq \alpha$. 
Recall that by definition $I(\sp(u))=[2^{i'}a',2^{i'}(a'+b')]$. We have, $a''-4\log n\leq a''-b'$ because $b'\leq 4\log n$. Since $I(u)\subseteq I(\sp(u))$, we also have $2^{i'}(a'+b')\geq \beta > \alpha \geq 2^{i'} a''$. Thus $a''-b'<a'$. Finally, again since $I(u)\subseteq I(\sp(u))$, we have $2^{i'}a'\leq \alpha$, and thus $a''\geq a'$. Combining the above inequalities, we get that $a'\in[a''-4\log n-1,a'']$. The marker now encodes the integer $t\in [0,4\log n-1]$ such that $a'=a''-t$. This is done in consuming another $\lceil\log\log n\rceil+2$ bits.
Hence, the following follows by construction:

\begin{lemma}
Given a tree $T\in\cT(n)$, the marker $\cM$ assigns labels to the nodes of $T$ in linear time, 
and each label is encoded using $\log n+6\lceil\log\log n\rceil+7$ bits.
\end{lemma}

%-----------------------------------------------------------
\subsection{The decoder $\cD$}
%-----------------------------------------------------------

Now, we describe our decoder $\cD$. Given the labels $L(u)$ and $L(v)$ assigned by $\cM$ to two different nodes in some tree $T$, the decoder $\cD$ needs to find whether $u$ is an ancestor of $v$ in $T$. (Observe that since each node receives a distinct label, the decoder can easily find out
if $u$ and $v$ are in fact the same node, and, in this trivial case, it simply outputs 0.)

The decoder inspects the first $\log n+3\lceil\log\log n\rceil+3$ least significant bits of $L(u)$ to extract  $I(u)$ (recall that $I(u)= I_{i,a,b}$ is encoded by storing explicitly the three parameters $i$, $a$, and $b$). Then,
once $I(u)=[\alpha,\beta]$ has been reconstructed from $L(u)$, the decoder
aims at extracting $I(\sp(u))$. For this purpose, it first reconstructs $i'$ and $b'$ that have been explicitly encoded in the next $2\lceil\log\log n\rceil+2$ bits.
Then, it
computes the largest integer $a''$ such that $2^{i'} a''\leq \alpha$.
 The decoder then proceeds by extracting $t$, and computes $a'=a''-t$. At this point, $\cD$ have reconstructed both $I(u)$ and $I(\sp(u))$.
Similarly, $\cD$ extracts $I(v)$ by inspecting $L(v)$.

Finally, the boolean decoder $\cD$ outputs~1 if and only if the two decoding conditions \D1 and \D2
hold at $u$ w.r.t. $v$  (see  Definition~\ref{def:decodingcond}).

\begin{lemma}
Let $L(u)$ and $L(v)$ be two labels assigned by $\cM$ to two nodes in $T$.
The decoder $\cD(L(u),L(v))$ performs in constant time, and satisfies $\cD(L(u),L(v))=1$ if and only if $u$ is an ancestor of $v$ in $T$.
\end{lemma}

\begin{proof}
The fact that $\cD(L(u),L(v))=1$ if and only if  $u$ is an ancestor of $v$ in $T$ follows from the fact that the intervals are assigned by the marker via a legal-interval mapping (cf. Corollary~\ref{corollary}). Since $I(u)= I_{i,a,b}=[2^ia,2^i(a+b)]$ with all three parameters $i$, $a$, and $b$ stored explicitly, computing $I(u)$ from $L(u)$ can be achieved in constant time. (Note that $2^ia$, for example, can be obtained from $a$ but a simple shift of $i$ bits.)
Similarly, $I(v)$ can be extracted from $L(v)$ in constant time. Computing $I(\sp(u))$ just needs a simple substraction and a division by a power of~2, which again amounts to a simple shift operation. The lemma follows.
\end{proof}

This completes the proof of Theorem~\ref{theo:main}.

%%%%%%%%%%%%%%%%%%%%%%%%%%%%%%%%%%%%
\section{Conclusion}
%%%%%%%%%%%%%%%%%%%%%%%%%%%%%%%%%%%%

Our ancestry labeling scheme is using labels of optimal size $\log_2n+O(\log\log n)$ bits, to the price of a decoding mechanism based of an interval condition slightly more complex than the simple interval containment condition. Although this has no impact on the decoding time (our decoder still works in constant time), the question of whether there exists an ancestry labeling scheme with labels of size $\log_2n+O(\log\log n)$ bits, but using solely the  interval containment condition, is intriguing.

\paragraph{Acknowledgments:} the authors are very thankful to Sundar Vishwanathan  and Jean-Sebastien Sereni for helpful discussions.

%\newpage

%%%%%%%%%%%%%%%%%%%%%%%%%%%%%%%%%%%%

\end{document}